\documentclass[conference,a4paper]{IEEEtran}

\usepackage{graphicx,cite,epsfig,amssymb,amsmath}
\usepackage{pifont}
\usepackage{stmaryrd}
\usepackage{bbding}
\usepackage{subfigure}
\usepackage{algorithm}
\usepackage{algorithmic}
\usepackage{slashbox,multirow}
\usepackage{mathrsfs}
\usepackage{latexsym}
\usepackage{indentfirst}
\usepackage{fmtcount}
\usepackage{cite}
\usepackage{url}

\allowdisplaybreaks[4]

\newtheorem{theorem}{Theorem}

\IEEEoverridecommandlockouts

\begin{document}
\title{Mobile Conductance and Gossip-based Information Spreading in Mobile Networks}

\author{\authorblockN{Huazi~Zhang$^{\dag\ddag}$, Zhaoyang~Zhang$^{\dag}$, Huaiyu~Dai$^{\ddag}$
\\$\dag$. Dept.of Information Science and Electronic Engineering, Zhejiang University, China.
\\$\ddag$. Department of Electrical and Computer Engineering, North Carolina State University, USA
\\Email: hzhang17@ncsu.edu, ning\_ming@zju.edu.cn, huaiyu\_dai@ncsu.edu}
\thanks{This work was supported in part by the National Key Basic Research Program of China (No. 2012CB316104), the National Hi-Tech R\&D Programm of China (No. 2012AA121605), the Zhejiang Provincial Natural Science Foundation of China (No. LR12F01002), the Supporting Program for New Century Excellent Talents in University (NCET-09-0701), and the US National Science Foundation under Grants ECCS-1002258 and CNS-1016260.}}
\maketitle

\begin{abstract}
In this paper, we propose a general analytical framework for information spreading in mobile networks based on a new performance metric, mobile conductance, which allows us to separate the details of mobility models from the study of mobile spreading time. We derive a general result for the information spreading time in mobile networks in terms of  this new metric, and instantiate it through several popular mobility models. Large scale network simulation is conducted to verify our analysis.
\end{abstract}

\section{Introduction}\label{intro}
\noindent Information spreading and sharing becomes an increasingly important application in current and emerging networks, more and more through mobile devices and often in large scales. Recently some interesting analytical results for information spreading in dynamic wireless networks have started to emerge (see \cite{ZKongMobiHoc2008,Spreading-MEG,Velocity-mobility,Gossip-mobility,Zhang-WCSP2012-nonrelay} and references therein). An observation is that, most existing analytical works focus on specific mobility models, in particular random-walk like mobility models. It is thus desirable to develop a more general analytical framework for information dissemination in mobile networks which can address different types of mobility patterns in a unified manner.

Information dissemination in static networks has been well studied in literature (see \cite{Gossip} and references therein), where an important result is that the spreading time is essentially determined by a graph expansion property, conductance, of the underlying network topology.
Conductance represents the bottleneck for information exchange in a static network, and this motivates us to explore its counterpart in mobile networks.
The main contributions of this paper are summarized below.

\begin{enumerate}
\item
Based on a ``move-and-gossip" information spreading model (Section \ref{problem-formulation}), we propose a new metric, mobile conductance, to represent the capability of a mobile network to conduct information flows (Section \ref{def}). Mobile conductance is dependent not only on the network structure, but also on the mobility patterns. Facilitated by the definition of mobile conductance, a general result on the mobile spreading time is derived for a class of mobile networks modeled as a stationary Markovian evolving graph (Section \ref{main}).

\item
We evaluate the mobile conductance for several widely adopted mobility models, as summarized in Table.~\ref{result-table} \footnote{We follow the standard notations. Given non-negative functions $f(n)$ and $g(n)$: $f(n) = O(g(n))$ if there exists a positive constant $c_1$ and an integer $k_1$ such that $f(n) \leq c_1 g(n)$ for all $n \geq k_1$;
$f(n) = \Omega(g(n))$ if there exists a positive constant $c_2$ and an integer $k_2$ such that $f(n) \geq c_2 g(n)$ for all $n \geq k_2$;
$f(n) = \Theta(g(n))$ if both $f(n) = O(g(n))$ and $f(n) = \Omega(g(n))$ hold.}. In particular, the study on the fully random mobility model reveals that the potential improvement in information spreading time due to mobility is dramatic: from $\Theta \left( {\sqrt n } \right)$ to $\Theta \left( \log n \right)$ (Section \ref{app}). We have also carried out large scale simulations to verify our analysis (Section \ref{simu}).
\end{enumerate}

\begin{table}[!t]
\renewcommand{\arraystretch}{1.1}
\caption{Conductance of Different Mobility Models} \label{result-table} \centering
\begin{tabular}{|c|c|}
\hline
Static Conductance & ${\Phi _s} = \Theta \left( {\sqrt {\frac{{\log n}}{n}} } \right).$ \\
\hline\hline
Mobility Model & Mobile Conductance ${\Phi_m}$ \\
\hline
Fully Random & $\Theta \left( 1 \right)$ \\
\hline
Partially Random & ${\left( {\frac{{n - k}}{n}} \right)^2}{\Phi _s} + \frac{{k\left( {2n - k} \right)}}{{2{n^2}}}$ \\
\hline
Velocity Constrained & $\Theta \left( \max \left( v_{\max},r \right) \right) $ \\
\hline
Area Constrained (1-d) & $\frac{n_V^2+n_H^2}{n^2}{\Phi _s} + {\frac{{{n_V}{n_H}}}{{n^2}}}$ \\
\hline
Area Constrained (2-d) & $\Theta \left( \max \left( r_c,r \right) \right) $ \\
\hline
\end{tabular}
\end{table}

\section{Problem Formulation}\label{problem-formulation}

\subsection{System Model}

\noindent We consider an $n$-node mobile network on a unit square $\Omega$, modeled as a time-varying graph $G_{t}\triangleq(V,E_t)$ evolving over discrete time steps. The set of nodes $V$ are identified by the first $n$ positive integers $[n]$. One key difference between a mobile network and its static counterpart is that, the locations of nodes change over time according to certain mobility models, and so do the connections between the nodes represented by the edge set $E_t$. The classic broadcast problem is investigated: one arbitrary node $s$ holds a message at the beginning, which is spread to the whole network through a randomized gossip algorithm. During the gossip process, it is assumed that each node \emph{independently} contacts one of its neighbors uniformly at random, and during each meaningful contact (where at least one node has the piece of information), the message is successfully delivered in either direction (through the ``push" or ``pull" operation) \cite{Gossip}.

\vspace{-0mm}
\begin{figure}[t] \centering
\includegraphics[width=0.40\textwidth]{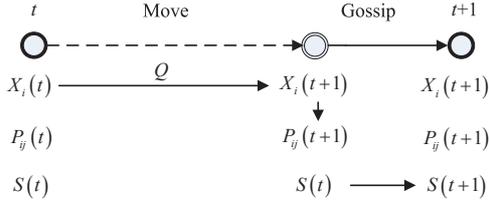}
\caption {Move-and-Gossip Spreading Model}
\label{Move-and-Gossip}
\end{figure}
\vspace{-0mm}

In contrast to the static case, there is an additional moving process mixed with the gossip process. In this study, we adopt a move-and-gossip model
as shown in Fig. \ref{Move-and-Gossip} to describe information spreading in a mobile network and facilitate our analysis. Specifically, each time slot is decomposed into two phases: each node first \emph{moves} independently according to some mobility model, and then \emph{gossips} with one of its \emph{new} neighbors. $X_i(t)$ denotes the position of node $i$, while $S(t)$ denotes the set of informed nodes (with $S\left( 0 \right) = \{s\}$), at the \emph{beginning} of time slot $t$. Note that $X_i(t)$ changes at the middle of each time slot (after the move step), while $S(t)$ is not updated till the end (after the gossip step). $P_{ij}(t+1)$ is used to denote the probability that node $i$ contacts one of its \emph{new} neighbors $j \in {\cal N}_i(t+1)$ in the gossip step of slot $t$; for a natural randomized gossip, it is set as $1/|{\cal N}_i(t+1)|$ for $j \in {\cal N}_i(t+1)$, and $0$ otherwise.

It is assumed that the moving processes of all nodes $\{X_i(t), t \in \mathrm{\mathbf{N}}\}$, $i \in [n]$, are independent stationary Markov chains, each starting from its stationary distribution with the transition distribution $q_i$, and collectively denoted by $\{\mathbf{X}(t), t \in \mathrm{\mathbf{N}}\}$ with the joint transition distribution $Q=\prod\limits_{i = 1}^n {{q_i}}$. While not necessary, we assume the celebrated random geometric graph (RGG) model \cite{RGG-book} for the initial node distributions for concreteness (particularly in Section \ref{app}), i.e., $G_{0}=G(n,r)$, where $r$ is the common transmission range, and all nodes are uniformly distributed. Under most existing random mobility models\cite{Mobility-throughput,ZKongMobiHoc2008,Spreading-MEG,Velocity-mobility,Gossip-mobility}, nodes will maintain the uniform distribution on the state space $\Omega$ over the time. The speed of node $i$ at time $t$ is defined by $v_i(t)=|X_i(t+1)-X_i(t)|$, assumed upper bounded by $v_{max}$ for all $i$ and $t$. We also assume that the network graph remains \emph{connected under mobility}; for RGG this implies $v_{\max}+r=\Omega(\sqrt{\log n/n})$\footnote{This requirement is already a relaxation as compared to $r=\Omega(\sqrt{\log n/n}))$ demanded for static networks. Actually our result only requires $E_{Q} \left[ {{N_{S'}}\left( {t + 1} \right)} \right]>0$; see (\ref{mobile-conductance-approx}).}.

\subsection{Mobility Model}\label{mobility-model}
\noindent The following mobility models are considered in this study:

Fully Random Mobility \cite{Mobility-throughput}:  $X_i(t)$ is uniformly distributed on $\Omega$ and i.i.d. over time. In this case, $v_{\max}=O\left(1\right)$. This idealistic model is often adopted to explore the largest possible improvement brought about by mobility.

Partially Random Mobility: $k$ randomly chosen nodes are mobile, following the fully random mobility
model, while the rest $n-k$ nodes stay static. This is one generalization of the fully random mobility model.

Velocity Constrained Mobility \cite{Velocity-mobility, Spreading-MEG}: This is another generalization of the fully random mobility model, with an arbitrary $v_{\max}$. In this case, $X_i(t+1)$ is uniformly distributed in the circle centered at $X_i(t)$ with radius $v_{\max}$.

One-dimensional Area Constrained Mobility \cite{Gossip-mobility}: In this model among the $n$ nodes, $n_V$ nodes only move vertically (V-nodes) and $n_H$ nodes only move horizontally (H-nodes). It is assumed that both V-nodes and H-nodes are uniformly and randomly distributed on $\Omega$, and the the mobility pattern of each node is ``fully random" on the corresponding one-dimensional path.

Two-dimensional Area Constrained Mobility\cite{ZKongMobiHoc2008,Two-dim-mobility}: In this model, each node $i$ has a \emph{unique} home point $i_h$, and moves around the home point within a disk of radius $r_c$ uniformly and randomly. The home points are fixed, independently and uniformly distributed on $\Omega$.

\section{Mobile Conductance}\label{def}
\noindent Conductance essentially determines the static network bottleneck in information spreading \cite{Conduct}. Node movement introduces dynamics into the network structure, thus can facilitate the information flows. In this work we define a new metric, \emph{mobile conductance}, to measure and quantify such improvement.


\emph{Definition:}
The mobile conductance of a stationary Markovian evolving graph with transition distribution $Q$ is defined as:

\begin{align}
 {\Phi _m}\left( Q \right)
&\triangleq \mathop {\min }\limits_{ \left| {S'\left( t \right)} \right| \leq n/2}  \left\{  E_{Q} \left( {\frac{{\sum\limits_{ i \in S'\left( t \right), j \in \overline {S'\left( t \right)} } {{P_{ij}}\left( {t + 1}\right)} }}{{\left| {S'\left( t \right)} \right|}}} \right) \right\} \label{mobile-conductance}\\ \label{mobile-conductance-approx}
&\overset{\text{(uniform)}}{\doteq} \mathop {\min }\limits_{ \left| {S'\left( t \right)} \right| \leq n/2}  \left\{ {\frac{{P\left(n, r \right)}}{{\left| {S'\left( t \right)} \right|}}E_{Q} \left[ {{N_{S'}}\left( {t + 1} \right)} \right]} \right\},
\end{align}
where ${S'\left( t \right)}$ is an arbitrary node set with size no larger than $n/2$, $P\left(n, r \right)$ is the common contact probability (in the order sense) for a RGG, and $N_{S'}\left( t+1 \right)$ is the number of connecting edges between ${S'\left( t \right)}$ and $\overline {S'\left( t \right)}$ after the move.


\emph{Remarks:}
\emph{1)} Some explanations for this concept are in order. Similar to its static counterpart, we examine the cut-volume ratio for an arbitrary node set ${S'\left( t \right)}$ at the beginning of time slot $t$. Different from the static case, due to the node motion ($X_i(t)\rightarrow X_i(t+1)$ in Fig. \ref{Move-and-Gossip}), the cut structure (and the corresponding contact probabilities $\{P_{ij}(t)\}$) changes. Thanks to the stationary Markovian assumption, its expected value (conditioned on ${S'\left( t \right)}$) is well defined with respect to the transition distribution $Q$. Minimization over the choice of ${S'\left( t \right)}$ essentially determines the bottleneck of information flow in the mobile setting.

\emph{2)}
For a RGG $G(n,r)$, the stochastic matrix $P(t)=\left[P_{ij}(t)\right]_{i,j=1}^{n}$ changes over time (in terms of connections) governed by the transition distribution $Q$ of the stationary Markovian moving process,  but the values of non-zero $P_{ij}(t+1)$'s remain on the same order given that nodes are uniformly distributed, denoted as $P(n,r)=\Theta\left(\frac{1}{{n\pi {r^2}}}\right)$. This allows us to focus on evaluating the number of connecting edges between ${S'\left( t \right)}$ and $\overline {S'\left( t \right)}$ \emph{after} the move: ${N_{S'}}\left( t+1 \right)= \sum_{i \in S'\left( t \right),j \in \overline {S'\left( t \right)} } {{1_{ij}}\left( j \in {\cal N}_i(t+1) \right)}$.\footnote{${{1_{ij}}\left( j \in {\cal N}_i(t+1) \right)}$ is the indicator function for the event that node $i$ and $j$ become neighbors after the move and before the gossip step in slot $t$.} Therefore for network graphs where nodes keep uniform  distribution over the time, mobile conductance admits a simpler expression (\ref{mobile-conductance-approx}).

\emph{3)} This definition may naturally be extended to the counterpart of $k$-conductance in \cite{Gossip}, with the set size constraint of $n/2$ in (\ref{mobile-conductance}) replaced by $k$, to facilitate the study of multi-piece information spreading in mobile networks.

\section{Mobile Spreading Time}\label{main}

\noindent The metric of interest for information dissemination is the $\epsilon$-spreading time, defined as:
\begin{equation}\label{spreading-time}
 T_{\mathrm{spr}} (\epsilon)
 \triangleq  \mathop {\sup}\limits_{s \in
V} \inf \left\{ {t:\Pr \left( {\left|S\left( t \right)\right| \ne n\left|
{S\left( 0 \right) = \left\{ s \right\}} \right.} \right) \le
\epsilon } \right\}.
\end{equation}
Based on the definition of mobile conductance, we have been able to obtain a general result for information spreading time in mobile networks.

\begin{theorem}\label{theorem-mobile-spreading-time}
For a mobile network with mobile conductance $\Phi_m(Q)$,
the mobile spreading time scales as
\begin{equation}\label{mobile-spreading-time}
T_{\mathrm{spr}}( \epsilon,Q) = O\left( {\frac{{\log n + \log
\epsilon ^{ - 1} }}{{\Phi _m(Q) }}} \right).
\end{equation}
\end{theorem}

\begin{proof}[Sketch of Proof]
We follow the standard procedure of the static counterpart (e.g. in \cite{Gossip}), with suitable modifications to account for the difference between static and mobile networks. Starting with $|S(0)|=1$, the message set $S(t)$ monotonically grows through the information spreading process, till the time $|S(t)|=n$ which we want to determine. The main idea is to find a good lower bound on the expected increment $|S(t+1)|-|S(t)|$ at each slot. It turns out that such a lower bound is well determined by the conductance of the network. Since the conductance is defined with respect to sets of size no larger than $n/2$, a two-phase strategy is adopted, where the first phase stops at $|S(t)|\leq n/2$. In the first phase, only the ``push" operation is considered for nodes in $S(t)$ (thus the upper bound on the spreading time is safe); while in the second phase, the emphasis is switched to the ``pull" operation of the nodes in $\overline {S(t)}$ (whose size is no larger than $n/2$). Since these two phases are symmetric, we will only focus on the first one.

In the first phase, for each node $j \in \overline {S\left(t\right)}$, define a random variable $\Delta_j\left(t\right)$. If at least one node with the message moves to the $j$'s neighboring area in slot $t$ and ``pushes" the message to $j$ in the gossip step, one new member is added to the message set. We let $\Delta_j{\left( {t + 1} \right)}=1$ in this case, and $0$ otherwise. In the following, we will evaluate the expected increment $|S(t+1)|-|S(t)|$ conditioned on $S(t)$. The key difference between the static and mobile case is that, there is an additional move step in each slot; therefore, the expectation is evaluated with respect to both the moving and gossip process. This is where our newly defined metric, mobile conductance, enters the scene and takes place of the static conductance. Specifically, due to the independent actions of nodes in $S(t)$ after the move, we have
\vspace{-0mm}
\begin{align*}
E\left[ {{\Delta_j}\left( {t + 1} \right)\left| {S\left( t \right)} \right.} \right]
=& E_Q \left[ {1 - \prod\limits_{i \in S\left( t \right)} {\left( {1 - {P_{ij}}\left( {t + 1} \right)} \right)} } \right] \\
\ge &E_Q \left[ {1 - \prod\limits_{i \in S\left( t \right)} {\exp \left( { - {P_{ij}}\left( {t + 1} \right)} \right)} } \right] \\
\ge& \frac{1}{2}E_Q \left[ {\sum\limits_{i \in S\left( t \right)} {{P_{ij}}\left( {t + 1} \right)} } \right],
\end{align*}

\noindent where the first and the second inequalities are due to the facts of $1 - x < \exp \left( { - x} \right)$ for $x \ge 0$ and $1 - \exp \left( { - x} \right) \ge \frac{x}{2}$ for $0 \le x \le 1$, respectively. Then
%
\begin{align}\label{set-increase}
&E\left[ {\left| {S\left( {t + 1} \right)} \right| - \left| {S\left( t \right)} \right|\left| {S\left( t \right)} \right.} \right]
=\sum\limits_{j \in \overline {S\left( t \right)} } {E\left[ {{\Delta_j}\left( {t + 1} \right)\left| {S\left( t \right)} \right.} \right]}\nonumber \\
\ge &\frac{1}{2}E_Q \left[ {\sum\limits_{i \in  {S\left( t \right)} ,j \in \overline {S\left( t \right)} } {{P_{ij}}\left( {t + 1} \right)} } \right]\nonumber \\
= &\frac{{\left| {S\left( t \right)} \right|}}{2} E_{Q} \left( {\frac{{\sum\limits_{ i \in S\left( t \right), j \in \overline {S\left( t \right)} } {{P_{ij}}\left( {t + 1}\right)} }}{{\left| {S\left( t \right)} \right|}}} \right)\nonumber \\
\ge &\frac{{\left| {S\left( t \right)} \right|}}{2}\mathop {\min }\limits_{ \left| {S'\left( t \right)} \right| \leq n/2}  \left\{  E_{Q} \left( {\frac{{\sum\limits_{ i \in S'\left( t \right), j \in \overline {S'\left( t \right)} } {{P_{ij}}\left( {t + 1}\right)} }}{{\left| {S'\left( t \right)} \right|}}} \right) \right\}\nonumber \\
= &\frac{{\left| {S\left( t \right)} \right|}}{2}{\Phi _m}\left( Q \right).
\end{align}
The form of \eqref{set-increase} is consistent with the counterpart in static networks \cite{Gossip}. Therefore, we can follow the same lines in the rest part of the proof.
\end{proof}

\section{Application on Specific Mobility Models}\label{app}
\noindent In the interest of space, the concept of mobile conductance is instantiated only through two mobility models in this section, and some less important technical details are omitted. The interested reader is referred to \cite{tech-report} for more details and results.

We will assume that the network instances follow the RGG model for concreteness, and evaluate (\ref{mobile-conductance-approx}). The main efforts in evaluation lie in finding the bottleneck segmentation (i.e., one that achieves the minimum in (\ref{mobile-conductance-approx})), and determining the expected number of connecting edges between the two resulting sets.
It is known \cite{Conduct} that for a static RGG $G(n,r)$, the bottleneck segmentation
is a bisection of the unit square, when $n$ is sufficiently large.
Intuitively, mobility offers the opportunity to escape from any bottleneck structure of the static network, and hence facilitates the spreading of the information. As will be shown below, fully random mobility destroys such a bottleneck structure, in that ${S'\left( t \right)}$ and $\overline {S'\left( t \right)}$ are fully mixed after the move; this move yields mobile conductance of $\Theta(1)$, a dramatic increase from static conductance  $\Theta \left( r \right)=\Theta \left( {\sqrt {\frac{{\log n}}{n}}} \right)$ \cite{Conduct}. Even for the more realistic velocity constrained model, part of the nodes from ${S'\left( t \right)}$ and $\overline {S'\left( t \right)}$ cross the boundary after the move and the connecting edges between the two sets are increased. The width of this contact region is proportional to $v_{\max}+r$.

\subsection{Fully Random Mobility}
\begin{theorem}\label{fully-random}
In fully random mobile networks, the mobile conductance scales as $\Theta \left( 1 \right)$.
\end{theorem}
\begin{proof}[Sketch of Proof]
 Since this mobility model is memoryless, for an arbitrary $S'\left( t
\right)$, the nodes in both $S'\left( t \right)$ and $\overline
{S'\left( t \right)}$ are uniformly distributed after the move, with density $|S'\left( t \right)|$ and $|\overline
{S'\left( t \right)}|$ respectively.
For each node in $S'\left( t \right)$, the size of its neighborhood
area is $\pi {r^2}$. Since each node contacts only one node in its radius, the expected number of contact pairs is
\begin{align}
 E_{Q} \left[ {{N_{S'}}\left( {t + 1} \right)} \right]
 = \left| {S'\left( t \right)} \right|\left| {\overline {S'\left( t \right)} } \right|\pi {r^2}.
\end{align}
Noting that
\begin{equation*}\label{full-mobility-conductance}
{\frac{{P\left(n, r \right)}}{{\left| {S'\left( t \right)}
\right|}}E_Q \left[
{{N_{S'}}\left( {t + 1} \right)} \right]}
 = \Theta \left( 1 \right),
\end{equation*}
regardless of the choice of $S'(t)$ (with size no larger than $n/2$) and $r$, we have $\Phi _m=\Theta(1)$. There is no bottleneck segmentation in this mobility model.
\end{proof}

\emph{Remarks: }In the gossip algorithms, only the nodes with the message can contribute to the increment of $\left| {S\left( t \right)} \right|$. Consider the ideal case that each node with the message contacts a node without the message in each step, which represents the fastest possible information spreading. A straightforward calculation \cite{tech-report} reveals that
$T_{\mathrm{spr}}(\epsilon) = \Omega \left( {\log n} \right)$ for an
arbitrary constant $\epsilon$. Theorem \ref{fully-random} indicates that in the fully random model, the corresponding mobile spreading time scales as $O\left( {\log n} \right)$ (when $\epsilon=O(1/n)$), so the optimal performance in
information spreading is achieved. The
potential improvement on information spreading time due to mobility
is dramatic: from $\Theta \left( {\sqrt n } \right)$ \cite{Gossip} to $\Theta
\left( \log n \right)$.

\subsection{Velocity Constrained Mobility}

\begin{theorem}\label{theorem-velocity}
For the mobility model with velocity constraint $v_{\max}$, the mobile conductance scales as $\Theta \left( \max \left( v_{\max},r \right) \right) $.
\end{theorem}
\begin{proof}[Sketch of Proof]
As argued in \cite{tech-report}, for the velocity constrained mobility model, the bottleneck segmentation is still the bisection of the unit square as shown in the upper plot of Fig.~\ref{velocity}, with $S'( t )$ on the left and $\overline {S'( t )}$ on the right before the move in time slot $t$. For better illustration, darkness of the regions in the figure represents the density of nodes that belong to
$S'( t )$. We can see that after the move, with
some nodes in both $S'( t )$ and $\overline {S'( t
)}$ crossing the border to enter the other half, a mixture
strip of width $2 \times v_{\max }$ emerges in the middle of the
graph.

We take the center of the graph as the origin. Denote $\rho _{S'( t )} ( l )$ and $\rho
_{\overline {S'( t )} } ( l )$ as the density
of nodes before moving, and $\rho
'_{S'( t )} ( l )$ and $\rho '_{\overline
{S'( t )} } ( l )$ as the density of nodes
after moving, with $l$ the horizontal coordinate\footnote{The node distributions are uniform in the vertical direction.}.
After some derivation, we have

{\footnotesize {\begin{equation*}
{\frac{{{{\rho '}_{S'(t)}}\left( l \right)}}{n} = \left\{ {\begin{array}{{ll}}
{1,} & {l <  - {v_{\max }}},\\
\begin{array}{l}
\arccos \left( {\frac{l}{{{v_{\max }}}}} \right)\\
 - \frac{l}{{{v_{\max }}}}\sin \left( {\arccos \frac{l}{{{v_{\max }}}}} \right),
\end{array} & {- {v_{\max }} < l < {v_{\max }}},\\
{0,} & {l > {v_{\max }}},
\end{array}} \right.}\\
\end{equation*}}}
and
\begin{equation*}
\frac{{{{\rho '}_{\overline {S'(t)}}}\left( l \right)}}{n} = 1 - \frac{{{{\rho '}_{S'(t)}}\left( l \right)}}{n}.
\end{equation*}
\vspace{-0mm}
\begin{figure}[h] \centering
\includegraphics[width=0.25\textwidth]{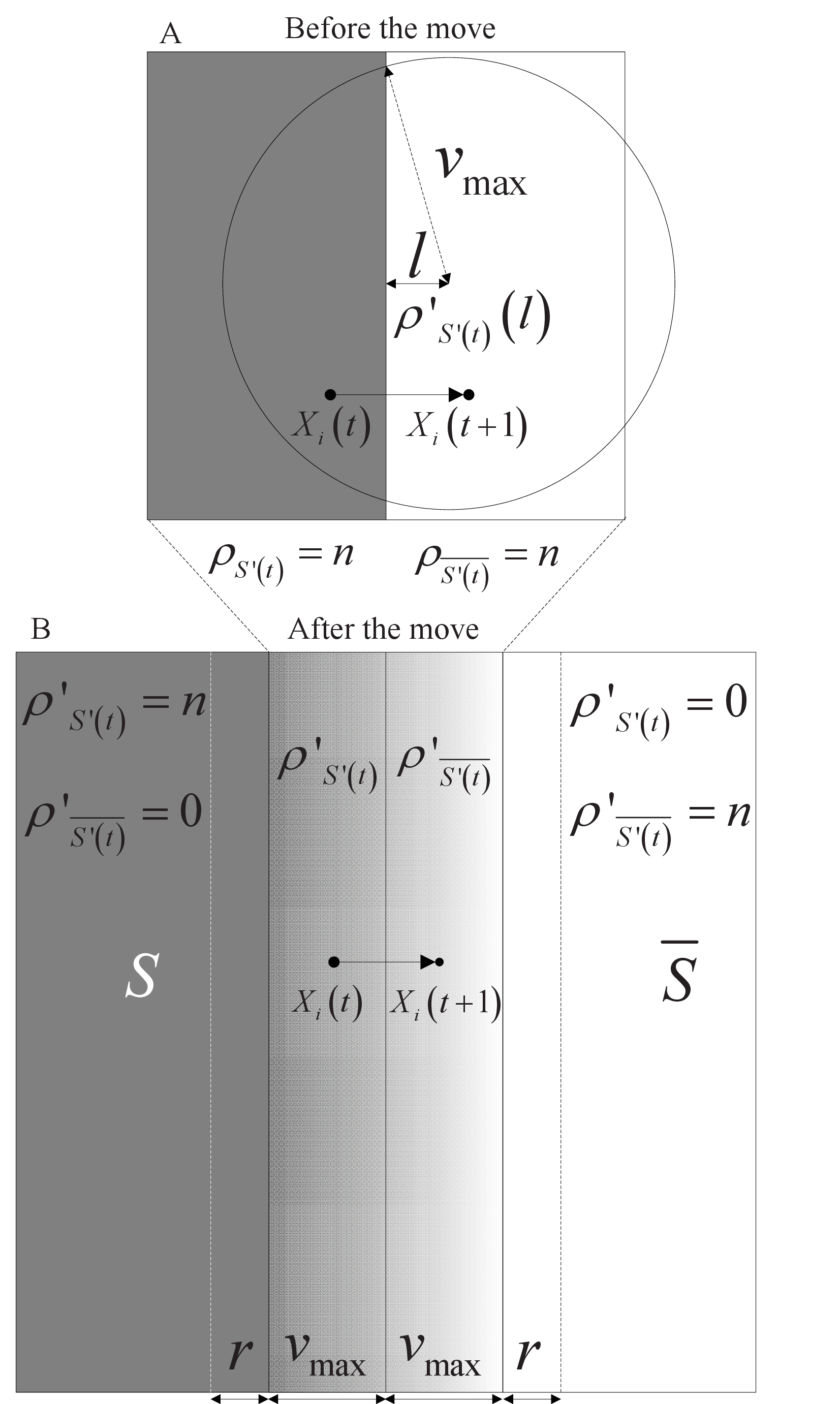}
\caption {Change of Node Densities Before and After the Move} \label{velocity}
\end{figure}

\vspace{-0mm}
\begin{figure}[h] \centering
\includegraphics[width=0.3\textwidth]{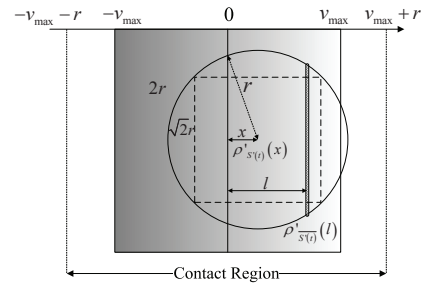}
\caption {Calculation of the Number of Contact Pairs} \label{velocity3}
\end{figure}

The contact pairs with the above bottleneck segmentation lie in the $2
\times \left( {v_{\max }  + r} \right)$ wide vertical strip in the center. All nodes outside this region will not contribute
to $N_{S'} \left( {t + 1} \right)$.
The number of contact pairs after the move can be calculated according to Fig.~\ref{velocity3}. The center of the circle with radius $r$ is $x$ away from the middle line. For node $i \in S'(t)$ located at the
center, the number of nodes that it can contact is equal to
the number of nodes belonging to $\overline {S'(t)}$ in the circle. Since the density of nodes belonging to $\overline {S'(t)}$ at positions $l$ away from the middle line
is ${\rho '_{\overline {S'(t)} } \left( l \right)}$, the number of
nodes that $i$ can contact is $\int\limits_{x - r}^{x
+ r} {\rho '_{\overline {S'(t)} } \left( l \right) 2 \sqrt {r^2  -
\left( {l - x} \right)^2 } dl}$. Taking all nodes belonging to $S'(t)$ in
the contact region into consideration, the expected number of
contact pairs after the move is
\begin{align}\label{calculus}
& E_Q \left[ {{N_{S'}}\left( {t + 1} \right)} \right]  \nonumber \\
  = & \int\limits_{ - v_{\max }  - r}^{v_{\max }  + r} {\rho '_{S'(t)} \left( x \right)\int\limits_{x - r}^{x + r} {\rho '_{\overline {S'(t)}} \left( l \right)2\sqrt {r^2  - \left( {l - x} \right)^2 } dldx} }.
 \end{align}
%
After some calculation\cite{tech-report}, the mobile conductance is well approximated by
\begin{align}\label{Velocity-conduct}
\Phi _m
\cong \left\{{\begin{array}{*{20}c}
  {\frac{1}{2}r + \frac{{v_{\max }^2 }}{{3r}},\quad \quad \quad \quad {\rm{for }}\quad v_{\max } \le \frac{1}{2} r }, \\
   { - \frac{{r^3 }}{{48v_{\max }^2 }} + \frac{{r^2 }}{{6v_{\max } }} + \frac{2}{3}v_{\max } ,\quad {\rm{for }}\quad v_{\max } > \frac{1}{2} r } . \\
\end{array}} \right.
\end{align}
\end{proof}

\emph{Remarks:} Theorem \ref{theorem-velocity} indicates that, when $v_{\max }=O(r)$, $\Phi _m = \Theta \left(
{r } \right)$, and the spreading time scales as $O(\log n/r)$, which degrades to the static case; when $v_{\max }=\omega(r)$, $\Phi _m= \Theta (v_{\max})$, and the spreading time scales as $O(\log n/v_{\max})$, which improves over the static case and approaches the optimum when $v_{\max}$ approaches $O(1)$. These observations are further verified through the simulation results below.

\section{Simulation Results}\label{simu}

\noindent We have conducted large-scale simulations to verify the correctness and accuracy of the derived theoretical results. In our simulation, up to 20,000 nodes are uniformly and randomly deployed on a unit square and move according to specified mobility models. The transmission radius $r\left( n \right)$ is set as $\sqrt {\frac{{8 \log n}}{\pi n}}$. For each curve, we simulate one thousand Monte-Carlo rounds and present the average.

The spreading time results for static networks and fully random mobile networks are shown in Fig.~\ref{Tspr_VelocityScaling} as the upper and lower bounds.
In particular, the bottommost curve (fully random mobility) grows in a trend of $\log n$ (note that the x-axis is on the log-scale), which confirms \emph{Theorem \ref{fully-random}}.
Fig.~\ref{Tspr_VelocityScaling} also confirms our remarks on \emph{Theorem \ref{theorem-velocity}}. When $v_{\max}=0.1$, the corresponding curve exhibits a slope almost identical to that for the fully random model. We also observe that $v_{\max}=\Theta \left( r \right)$ is a breaking point: lower velocity ($v_{\max}=o(r)=\Theta \left(\sqrt{\frac{1}{n}}\right)$) leads to a performance similar to the static case.

\vspace{-0mm}
\begin{figure}[h] \centering
\includegraphics[width=0.39\textwidth]{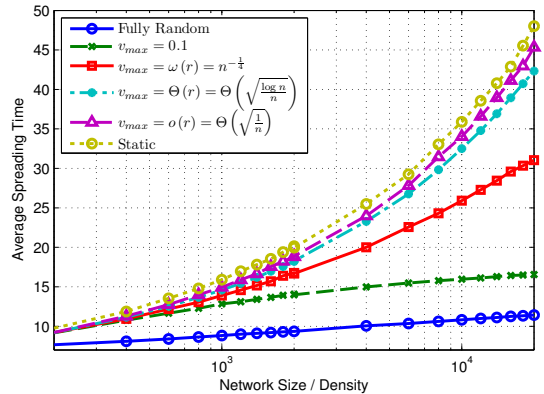}
\caption {Spreading Time under the Velocity-Constrained Mobility Model}
\label{Tspr_VelocityScaling}
\end{figure}

\section{Conclusions and Future Work}\label{conclusion}

\noindent In this paper, we analyze information spreading in
mobile networks, based on the proposed move-and-gossip information spreading model. For a dynamic graph that is connected under mobility, i.e., $v_{\max}+r=\Omega(\sqrt{\log n/n})$, we have derived a general expression for the information spreading time by gossip algorithms in terms of the newly defined metric mobile conductance, and shown that mobility can significantly speed up information spreading. This common framework facilitates the investigation and comparison of different mobility patterns and their effects on information dissemination.

In our current definition of mobile conductance, it is assumed that in each step, there exist some contact pairs between ${S'\left( t \right)}$ and $\overline {S'\left( t \right)}$ after the move. In extremely sparse networks (depending on the node density and transmission radius), we may have $E_{Q} \left[ {{N_{S'}}\left( {t + 1} \right)} \right]=0$. Let $T_{m}(i,j)\triangleq \inf \{t: j \in {\cal N}_i(t+1) \}$ be the first meeting time of nodes $i$ and $j$. We plan to extend the definition of mobile conductance to the scenario with $E[T_{m}(i,j)]<\infty$.

%

\end{document}